\documentclass[12pt]{article}
\usepackage{multicol, eurosym, sgame}
\usepackage[T1]{fontenc}
\usepackage{multirow,array,caption}
\usepackage{lmodern} 
\usepackage{geometry,kpfonts}
\usepackage{tikz}
\usepackage{hyperref}
\DeclareFontFamily{OT1}{pzc}{}
\DeclareMathAlphabet{\mathpzc}{OT1}{pzc}{m}{it}

\usepackage{setspace,graphicx}
\usepackage{amsmath,amsthm}
\usepackage{amssymb}
\usepackage{comment}
\usepackage{subcaption}  
\usepackage{ulem,bbm}
\usepackage{stackengine,cases}

\usepackage[authoryear]{natbib}

\newtheorem{proposition}{Proposition}
\newtheorem{lemma}{Lemma}
\newtheorem{definition}{Definition}

\newtheorem{claim}{Claim}

\title{Legitimacy of collective decisions: \\ a mechanism design approach\footnote{We thank Jeff Ely, Yukio Koriyama, Jean-François Laslier, Hervé Moulin, Tom Palfrey, Harry di Pei, Marzena Rostek, Fedor Sandomirskiy, Omer Tamuz, Dimitrios Xefteris, Yves Le Yaouanq and Siyan Xiong for their useful remarks and comments as well as conference participants at the Asian School in Economic Theory of the Econometric Society (Singapore 2022) and Université Paris Dauphine. This research is supported by two grants of the French National Research Agency (ANR) "Investissements d'Avenir": LabEx Ecodec/ANR-11-LABX-0047 and ANR-18-EURE-0005 / EUR DATA EFM.}}
\author{Margarita Kirneva\footnote{margarita.kirneva@polytechnique.edu. CREST \& Ecole Polytechnique} \: and Matías Nuñez\footnote{matias.nunez@polytechnique.edu. CREST, CNRS \& Ecole Polytechnique}}
\date{August 2023}

\begin{document}


\maketitle
\begin{abstract}
We design two mechanisms that ensure that the majority preferred option wins in all equilibria. The first one is a simultaneous game where agents choose other agents to cooperate with on top of the vote for an alternative, thus overcoming recent impossibility results concerning the implementation of majority rule. The second one adds sequential ratification to the standard majority voting procedure allowing to reach the (correct) outcome in significantly fewer steps than the widely used roll call voting. Both mechanisms use off-equilibrium lotteries to incentivize truthful voting. We discuss different extensions, including the possibility for agents to abstain.
\end{abstract}



\section{Introduction}

In legislative, referenda, and committee settings, majority voting is commonly used and is based on simple and intuitive axioms. This method plays a crucial role in analyzing democratic institutions; moreover it is particularly simple to use since voting sincerely is a weakly dominating strategy. However, there is the question of the legitimacy of the decision. For instance, if few voters take part in the election, the outcome might not reflect the preferences of the electorate and can be considered  illegitimate, as argued by \cite{qvortrup2005comparative} in the context of referenda.\footnote{Moreover, surveys find that voters often express regret for not participating in the election (see \cite{blais2019my} for a recent contribution). } Practitioners often introduce participation quorums to ensure that a sufficient number of voters participate, raising the likelihood that the collective decision is legitimate. The U.S. Senate, the U.K House of Commons, the U.N. General assembly and many other organizations use these quorums. However, the literature suggests that quorums may modify the incentives that voters face, leading to potential negative effects as the existence of quorum-busting, where the minority abstains to prevent the quorum to be reached (see \cite{herrera2010quorum} in an equilibrium model and by \cite{aguiar2016experimental} in an experimental setting).

Since the effectiveness of participation quorums to raise legitimacy is dubious, we consider the legitimacy issue through another perspective: the lens of implementation theory. The objective is to design voting mechanisms where all equilibria select some desirable alternative: an implementation problem à la \cite{maskin1999nash}. The situation involves two options, labeled as $a$ and $b$, which a group of voters must choose between. Each voter strictly prefers one option over the other. The goal of the designer is to select the option specified by a social choice rule (such as majority rule), without knowing the preference profile. However, voters are not forced to vote truthfully, so the designer aims to have a mechanism that ensures that the desired option wins in all equilibria of the game. \cite{xiong2021designing} demonstrates the existence of a two-alternative implementation problem. This problem states that, with two alternatives, a social choice  rule (among which majority rule is the most salient one) is implementable via a simultaneous voting mechanism\footnote{A voting mechanism allows each agent to vote for each option and is monotonic in the usual sense: if $x$ wins at some profile and gets additional support from some agent, $x$ remains the winner.} if and only if it is dictatorial. The category of voting mechanisms is large and covers most of the currently used ones.  Our contribution is the design of two classes of mechanisms that circumvents this implementation problem.

Our first contribution is the design of the Bloc formation mechanism, the first simultaneous one, different than an integer game, that Nash implements majority rule. The mechanism is not a pure voting mechanism because it requires each agent to vote for one of the two options and to nominate $p$ agents. The outcome depends on whether a majority group of voters vote for the same alternative, say $x$, while nominating only voters in the group. In such case, a bloc in favor of $x$ is formed, and $x$ wins. If no bloc is formed, the outcome is an endogenous lottery that incentivizes truth-telling. This mechanism's definition actually allows us to see a voting profile as a directed network, making the outcome of the vote depending on the network structure generated by the voting profile. The implementation result extends to a setting with an even number of voters.

Our second contribution deals with sequential voting mechanisms. Remark that sequential majority voting (or roll call voting) subgame perfect implements majority rule when votes are mandatory. This is no longer the case under a participation quorum due to quorum busting. Moreover, this system  can be lengthy as the number of steps required goes from $p+1 $ (the first voters all agree) to $2p+1$ ($p+1$ prefer one option and $p$ the other). Our contribution is the design of the Majority with Random confirmations, or RC mechanism that reduces the numbers of steps. Each voter votes for one of the two options (the voting stage). Half plus one of the voters are (randomly) selected, and each
one, consecutively, declares whether or not he approves of the winner of the vote ( the confirmation stage). If one
of them approves, the procedure terminates and that option is elected. If all
disapprove, a lottery is cast between the two options, where the probability
of each option is proportional to the votes it received in the voting stage.  The main advantage of the RC mechanism with respect to the sequential majority one is that the number of steps required is at most $p+1$ and might be smaller (since it terminates with the first approval declaration
from the $p + 1$ selected people). We extend the mechanism and the implementation result in two directions: (i) one where agents can abstain in any of the stages making the strategic problem richer and (ii) a setting with incomplete information where agents do not know the preferences of the rest of agents. Result $(i)$ is important since it shows that there is no need of introducing participation quorums to ensure that the outcome is correct. Result $(ii)$ shows that the implementation via RC mechanism does not depend on complete information. 
As a final result, we extend the logic of the RC mechanism beyond majority rule, as any supermajority rule can be implemented via a simple modification of the RC mechanism.

This work is organized as follows. After laying out the model in Section \ref{sec:setting}, Section \ref{sec:altmechanisms} considers the Bloc formation mechanism. Section  \ref{sec:majorityform}  analyzes the Majority voting with Random confirmations and its different extensions and Section \ref{sec:conclusion} concludes. The appendix contains the proof of the implementation via the BF mechanism as well as the analysis of three extensions of the RC mechanism (abstention, shares revelation and even number of voters).

\color{black}

\section{Review of the literature}

This paper belongs to a new strand of the literature in implementation where the aim is to find attractive implementing mechanisms that could be tested experimentally rather than analyzing whether a social goal is implementable as in the classic strand following \cite{maskin1999nash}. Among the papers in this literature, we could cite the recent contribution by \cite{chen2023getting} showing that any social choice rule is implementable via two-stage mechanisms involving transfers and lotteries and the one by \cite{echenique2022price} which proves that a two-stage mechanism involving prices suffices to implement efficient outcomes.

Our focus is on mechanisms with off-equilibrium lotteries. These mechanisms are known to be more permissive than deterministic mechanisms. See \cite{benoit2008nash}, \cite{bochet2007nash} and \cite{laslier2021solution} for recent contributions.

\cite{borgers2014robust} develop a similar idea to show that one
can achieve Pareto improvements over random dictatorship through simultaneous mechanisms. In the related framework of the Condorcet jury theorem, \cite{laslier2013incentive} proposes the introduction of a "slightly randomized majority rule" to ensure that the unique equilibrium is informative. Our mechanism shares a similar spirit to this idea but without randomization in equilibrium. Likewise, \cite{nunez2019truth} and \cite{azevedo2019strategy} describe similar ideas for large populations of agents. See \cite{moore1988subgame} and \cite{vartiainen2007subgame} for an analysis of rules which are subgame perfect implementable as well as \cite{aghion2012subgame} for the robustness of subgame perfect implementation to information perturbations.

\section{Model \label{sec:setting}}

We consider a finite set $I=\{1,\ldots,n\}$ of agents, with generic element $i$, who need to choose one option out of the set $A=\{a, b\}$ with generic element $x$. We assume that $n$ is odd ($n=2p+1$) except in Section \ref{sec:evenBF} and in Appendix \ref{section:rceven}.  Each agent has strict and complete preferences over $A$ where $a R_i b$ denotes that $a$ is strictly preferred to $b$. A vector $R=(R_1,\ldots, R_n) \in \mathcal{R}^n$ denoted the preference profile where $\mathcal{R}$ is the set of strict preference relations over $A$. A social choice function (SCF) is a mapping $f: \mathcal{R}^n \rightarrow A$ that selects a single option for each profile $R$. The majority rule, denoted $Maj$, is the SCF  that selects the majority preferred option : for each preference profile $R$,

\begin{equation}
   Maj(R)=
    \begin{cases}
        a & |\{i \in I : a R_i b\}| \geq p+1 \text{ and },\\
        b & \text{otherwise.}
    \end{cases}
\end{equation}

We let $\Delta$ denote the set of lotteries over $A$ with $\Delta=\{\beta:A \rightarrow [0,1] : \sum_{A}\beta(a)=1\}$. A simultaneous \textit{mechanism} is a function $g: M \rightarrow \Delta$ that assigns to every $m\in M$ a unique element of $\Delta$, where $M=\prod_{i\in I}M_i$, and $M_i$ is the strategy space of agent $i$.

We assume that preferences over lotteries satisfy stochastic dominance (\textbf{SD}). In our setting, \textbf{SD} requires that an agent (weakly) prefers lottery $\beta$ over lottery $\eta$ if and only if $\beta$ assigns (weakly) higher probability to her preferred option $x$: :

$$\beta \tilde{R}_i^{\textbf{SD}} \eta \Longleftrightarrow \beta(x) \geq  \eta(x) \text{ and } \beta R_i^{\textbf{SD}} \eta \Longleftrightarrow \beta(x) >  \eta(x),$$

where $\beta \tilde{R}_i^{\textbf{SD}} \eta$ means that agent $i$ weakly prefers $\beta$ to $\eta$ and $\beta R_i^{\textbf{SD}} \eta$ implies that she strictly prefers the former to the latter. This definition implies that a lottery $\beta$  stochastically
dominates lottery $\eta$ when $\beta$ yields at least as much expected utility
as $\eta$ for any von-Neumann Morgenstern utility representation consistent with the ordinal preferences.

\subsection{Implementation notions}

A simultaneous mechanism specifies a game-form: this means that, when the mechanism is coupled with preferences over options for each of the agents, it defines a normal-form game. A Nash equilibrium of the mechanism $g$ is a profile $m\in M$ such that $g(m)\tilde{R}_i^{\textbf{SD}}g(m'_i,m_{-i})$ for each $i\in I$ and any $m'_i\in M_i$. For a mechanism $g$, let $\mathrm{NE}^{g}(R)$ denote the set of Nash equilibria at preference profile $R$. A mechanism Nash implements a social choice function $f$ if for any $R$, the outcome of any member of $\mathrm{NE}^{g}(R)$ is an element of $f(R)$ and any element of $f(R)$ is the outcome of some member of $\mathrm{NE}^{g}(R)$.

A sequential mechanism is an extensive game form $\Gamma=(\mathcal{H},M,\mathcal{Z},g)$ where $\mathcal{H}$ is the set of all histories, $M=M_1\times \ldots M_n$ is the message space with $M_i=\prod_{h\in \mathcal{H}}M_i(h)$ for all $i$ where $M_i(h)$ is the set of available messages for $i$ at history $h$; $\mathcal{Z}$ describes the history that immediately follows history $h$ given that $m$ has been played; and $g$ is the outcome that maps the set of terminal histories. The notation $g(m;h)$ denotes the outcome that obtains when agents use strategy profile $m$ starting from history $h$.

There is an initial history $\emptyset\in \mathcal{H}$ and $h_t=(\emptyset,m^1,m^2,\ldots,m^{t-1})$ is the history at the end of period $t$, where for each $k$, $m^k \in M(h_k)$. If for $t'\geq t+1$, $h_{t'}=(h_t,m^t,\ldots,m^{t'-1})$, then $h_{t'}$ follows history $h_t$. Since $\Gamma$ contains finitely many stages, there is a set of terminal histories $H_T\subset \mathcal{H}$ such that $H_T=\{h\in \mathcal{H}: \text{ there is no $h'$ following $h$}\} $. A subgame-perfect equilibrium for the game $\Gamma(R)$ is an element $m\in M$ such that, for each agent $i$, $g(m;h)\tilde{R}_i^{\textbf{SD}}g(m'_i,m_{-i};h)$ for all $m'_i\in M_i$ and all $h\in \mathcal{H}\setminus H_T$. The set $SPNE^g(R)$ denotes the set of subgame-perfect equilibria of the game $\Gamma(R)$. We say that a mechanism implements the $SCR$ $f(R)$ in subgame-perfect equilibria, if for each $R$, the outcome of any member of $SPE(\Gamma(R))$ is an element of $f(R)$ and any element of $f(R)$ is the outcome of some member of $SPE(\Gamma(R))$.

 A similar idea applies to the incomplete information setting where the equilibrium notion upon which we rely is Perfect Bayesian equilibrium denoted PBE in the sequel (a formal definition is skipped \footnote{For the formal definition see \cite{fudenberg1991perfect}}).

\subsection{A discussion on the Majority mechanism}

The following two mechanisms are relevant in both theory and practice, as discussed in the introduction. The majority mechanism, denoted $\theta_M:A^n \rightarrow A$, requests each agent to vote for one of the two options and selects $Maj(m)$ as the winner. The majority mechanism with quorum $Q$, denoted $\theta_{Q}:(A \cup \text{abs} )^n \rightarrow A$, requires that each voter either announces their vote for an option in $A$ or abstains. For each profile $m$, the outcome $\theta_{Q}(m)$ can be expressed as:
\begin{enumerate} 
    \item if $n_{abs}\leq Q$, $a$ wins if $n_a \geq \frac{n-n_{abs}}{2}$, $b$ wins otherwise,
 \item if $n_{abs}\geq Q$, $b$ is the winner.
    \end{enumerate}
 where  $n_a$, $n_b$, and $n_{abs}$ respectively represent the number of votes for $a$, $b$, and abstentions

As we now discuss, both mechanisms do not implement the majority-preferred option.
 
Remark first that the mechanism $\theta_M$ fails to Nash implement the majority rule. Since the mechanism is strategy-proof, it has an equilibrium in sincere strategies where $Maj(R)$ is the winner. Although this equilibrium is focal, the mechanism has many other equilibria for each preference profile $R$, many of which do not select $Maj(R)$. For instance, the strategy profile $m$ with $m_i=b$ for all $i\in I$ is an equilibrium for any $R$ since no player can prevent the victory of $b$. While theoretically possible, there is a widely held belief that such equilibria rarely arise in practice: indeed, with the mechanism $\theta_M$, some agents in the majority need to vote for their worst-preferred option to allow for the defeat of the majority option. 

A similar logic to the one with the Majority mechanism shows that the mechanism $\theta_Q$ fails to implement $Maj(R)$. However, this failure to implement the majority rule is more credible than with $\theta_M$. Indeed, the majority winner of the recorded votes may fail to select the majority preferred option of the electorate since every voter decides whether to participate. This means that there are equilibria where every participating agent votes sincerely while $Maj(R)$ loses.

Regarding implementation via sequential mechanisms, remark that the dynamic counterpart of $\theta_M$ does implement majority rule in subgame-perfect equilibria. Indeed, voters correctly anticipate the moves of the successors which ends up in the correct option being selected. However, the dynamic $\theta_Q$ (much more used in practice than $\theta_M$) does not follow the same logic since adopting the strategy of not showing-up the minority of voters can prevent the victory of the majority-preferred option.

\section{A simultaneous mechanism\label{sec:altmechanisms}}

In this section, we introduce the Bloc formation mechanism (BF mechanism), the first mechanism, beyond integer games, that Nash implements the majority rule. We also comment about its interpretation as a network  formation game and show its strategic behavior with an even number of agents.

\subsection{Simultaneous blocs\label{simbloc}}

In the BF mechanism, the message $m_i$ of agent $i$ consists of (1.) a vote for an option $v_i$ and (2.)  a nomination of $p$ agents excluding herself (denoted $c_i$). Formally, the mechanism is denoted $\chi_{BF}: M\rightarrow \Delta$ with, for all $i\in I$, $M_i:=A \times 2_p^{-i}$ where $2_p^{-i}$ denotes all the sets of $p$ agents different from $i$. 

\vspace{.3cm}

The central notion of this mechanism is the idea of a bloc of agents. For each option $x$, a bloc in favor of option $x$  is a majority group of agents, denoted $B$, such that each agent votes for $x$ while nominating only agents in $B$. This can be formally defined as follows.
\begin{definition}\label{def:bloc}
For each $x\in A$, any set $B$ of agents with $|B|\geq p+1$ forms a \textbf{bloc} in favor of option $x$ in the profile $m$ if:
\begin{enumerate}
    \item $v_i=x$ $\forall i\in B$ (only votes for $x$) and,
    \item  $c_i\subset B$  $\forall i\in B$ (only nominations in $B$).
\end{enumerate}
\end{definition}
\vspace{.3cm}

The outcome of the BF mechanism depends on whether the message profile has a bloc. Denote by $B^m$ the set of blocs formed in profile $m$.  By definition, all blocs in a profile (if any) favor the same option since each bloc contains a majority of agents. Therefore, for any profile $m$ in which there is a bloc in favor of option $x$, $\chi_{BF}(m)=x$.

\vspace{.3cm}

\noindent If the profile $m$ does not contain a bloc, the outcome is the lottery $\eta(m)$ over $A$ with, for each $x\in A$:

$$\eta^x(m)=\sum_{i\in I}\eta_{i}(m) \mathbbm{1}\{v_i=x\} \: \text{ with } \eta_i(m)=\frac{|\{j\in I\setminus \{i\} : i \in c_j\}|}{np}.$$

To see the logic behind this formula, we let $\eta_i(m)$ be the weight of agent $i$, that is the share of nominations of $i$ in the total nominations $np$. By construction, $\sum_{i\in I}\eta_{i}(m)=1$ for any $m\in M$. When all the other agents nominate $i$, agent $i$ has the maximal possible weight of $\eta_i(m)=\frac{n-1}{np}$ whereas $\eta_i(m)=0$ when none of the other agents nominate $i$.

We thus interpret $\eta^x(m)$  as the sum of the weights of the agents who vote for $x$ so that, by construction, $\eta^a(m)+\eta^b(m)=1$. Notice that the weight $\eta^x(m)$ is strictly increasing in the number of nominations for agents voting for $x$ and, thus, in the number of agents voting for $x$ among nominated agents.

The previous rules of the mechanism can be summarized as follows. For each message profile $m$, the outcome of the mechanism $\chi_{BF}$ coincides with:
$$\chi_{BF}(m)=\begin{cases}
a & \text{if } m \text{ admits a bloc in favor of } a, \\ 
b & \text{if } m \text{ admits a bloc in favor of } b, \\
\eta(m) & \text{otherwise.}
\end{cases}$$

\medskip

To conclude the description of the BF mechanism notice that it is strategy-proof, that is for any agent $i$ with $a R_i b$ (resp. $bR_i a$), any nomination $c_i\in 2^{-i}_p$ and any message $m_{-i}$, agent $i$ weakly prefers to vote for $a$ (resp. $b$) since:

$$\chi_{BF}(a,c_i,m_{-i})\tilde{R}^{SD}_i\chi_{BF}(b,c_i,m_{-i}).$$

In what follows, we show that in any equilibrium of the BF mechanism, most agents strictly prefer to vote truthfully, ensuring that the majority-preferred option wins. This majority of voters that strictly prefer to vote honestly represents the main advantage of the BF mechanism for the usual majority voting one where, in some equilibria, all voters may be indifferent between their two votes.

\subsection{Voting profile as a directed graph\label{blocnetwork}}

It is useful to consider blocs in terms of the graph theory. Notice that for any message profile $m=(v,c)$ the nomination profile $c$ creates a directed graph in which the vertices are the agents and the edges their nominations. Formally, for each message profile $m=(v,c)$ denote by $G_{m}=(I,C)$ the directed graph formed by $c$ where the set $I$ of agents coincides with the set of vertices, and $C$ is the adjacency matrix such that $C_{ij}=1$ if $j \in c_i$ and $C_{ij}=0$ otherwise.

We can formulate an option definition of bloc using the adjacency matrix.

\begin{definition}\label{def:blocgraph}
For each $x\in A$, a set $B$ of agents with $|B| \geq p+1$ forms a \textbf{bloc} in favor of  option $x$ in profile $m$ if:

\begin{enumerate}
    \item  $v_i=x$ $\forall i \in B$ and,
    \item the restriction of $C$ to the set $B$, denoted $C_B$, is such that $\sum_{h\in B} C_{ih}=p$ $\forall i \in B$. 
\end{enumerate}

\end{definition}

It follows from Definition \ref{def:blocgraph} that if a set $B$ of agents forms a bloc in favor of $x$, then there is no path from any agent $i \in B$ to any agent $j \in I \setminus B$ in the associated graph $G_m$.  Remark that Definition \ref{def:bloc} is equivalent to Definition \ref{def:blocgraph}. Firstly, both definitions require at least $p+1$ agents to vote for the same option. To show that $\sum_{h\in B} C_{ih}=p$ $\forall i \in B$ is equivalent to agents in $B$ voting only for other agents in $B$, observe that the row $i$ of the adjacency matrix $C$ gives the nominations of agent $i$. Thus, $\sum_{h\in I} C_{ih}=p$ follows by definition. Since, according to Definition \ref{def:blocgraph}, $\sum_{h\in B} C_{ih}=p$ $\forall i \in B$, $C_{ih}=0$ for all $h \in I \setminus B$, that is no agent $i \in B$ is nominating an agent outside $B$ which proves the equivalence.

We need some additional definitions to formulate the main results regarding blocs.

We say that $G_J$ is a subgraph of a graph $G$ induced by the set $J \subseteq I$ of vertices if it includes all vertices in $J$ and its adjacency matrix $C_J$ is the restriction of $C$ to
$J$ (i.e includes only rows and columns corresponding to vertices in $J$).

\begin{definition}
A subgraph $G_J$ for some $J \subseteq I$ of a graph $G$ is strongly connected if there exists a path in each direction between any pair $i$, $j$ of vertices with $i,j \in J$.
\end{definition}

\begin{definition}
A bloc $B \subseteq I$ in favor of $x$ is effective iff it is strongly connected.
\end{definition}

 Each vertex in Figure \ref{fig:blocs} represents an agent, the letters within represent their votes, and the arrows indicate nominations. In Figure \ref{subfig:nobloc} , there are no blocs in the profile. The only potential bloc involves agents $\{1, 2, 3\}$, as they all vote for $a$ while the others vote for option $b$. Agent 1 and agent 3 nominate agent 5 and agent 4, respectively, which violates the conditions required to form a bloc.
 
  In Figure \ref{subfig:bloc}, the profile admits two blocs:   $\{1,4,5\}$ and $\{1,2,3,4,5\}$. Indeed, in both of these subsets all agents vote for $b$ and nominate only agents within the bloc. However, only the bloc $\{1,4,5\}$ is effective because it is strongly connected. Notice that there is no path from agent 4 to agent 2, which prevents $\{1,2,3,4,5\}$ from being an effective bloc.

\begin{figure}[h]
\center
\begin{subfigure}[b]{0.45\textwidth}
\centering
  \includegraphics[width=0.8\textwidth]{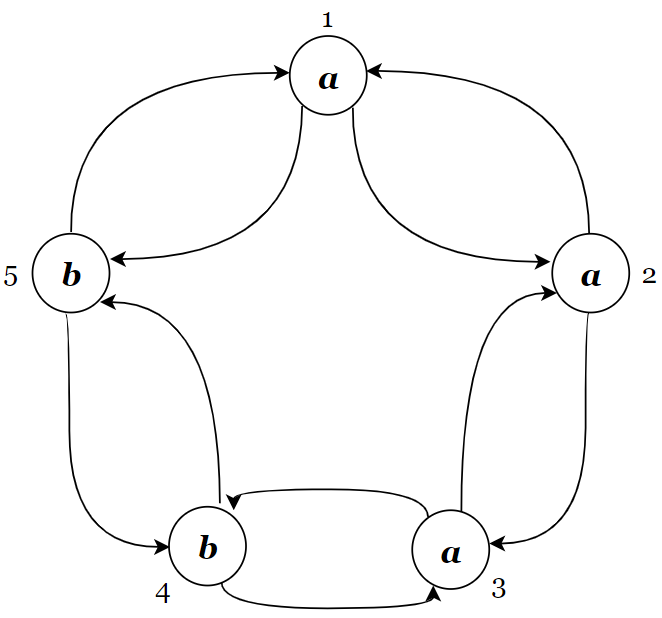}
  \caption{No blocs.}
  \label{subfig:nobloc}
\end{subfigure}
\hfill
\begin{subfigure}[b]{0.45\textwidth}
\centering
     \includegraphics[width=0.8\textwidth]{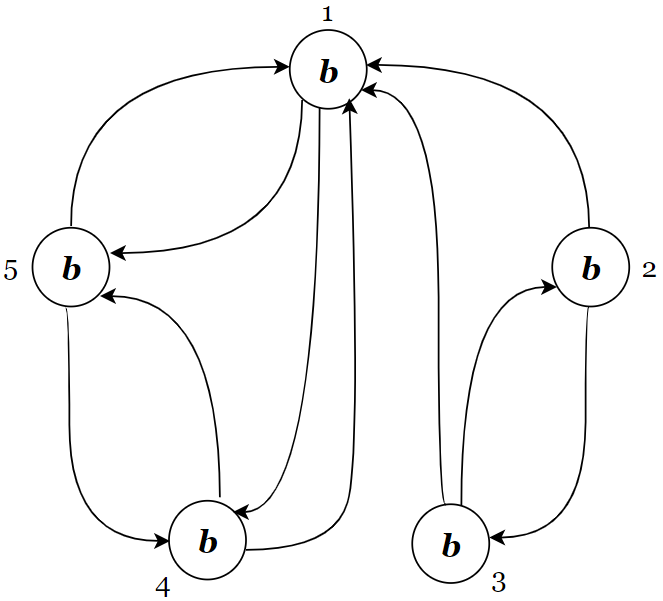} 
     \caption{Two blocs, $\{1,4,5\}$ and $\{1,2,3,4,5\}$, in favor of $b$.}
  \label{subfig:bloc}
\end{subfigure}
\caption{Voting profiles formed by the BF mechanism.}
\label{fig:blocs}
\end{figure}

The next proposition shows the existence and uniqueness of effective blocs.

\begin{lemma}\label{prop:effective}
Any profile $m$ admitting a bloc also admits an effective bloc $B^*$. Moreover, the effective bloc $B^*$ is unique and satisfies $B^*=\cap_{B \in B^{m}}B$.
\end{lemma}

\begin{proof}
\textbf{Existence.} Let $m$ be some profile with $B^{m} \neq \emptyset$ and consider w.l.o.g. that all blocs are in favor of $a$. Assume, for the sake of contradiction, that there is no effective bloc in $m$. This means that any bloc $B \in B^m$, is not effective and therefore not strongly connected. It follows that there are 2 vertices, namely $i$ and $j$,  with no path from $i$ to $j$, from $j$ to $i$, or in both directions. W.l.o.g. assume that there is no path from $i$ to $j$. It follows that we can find a partition $(B',B\setminus B')$ of $B$ such that: (1) $i \in B'$, (2)
there is no path from any agent in $B'$ to agent $j$ and (3) there is a path from any agent in $B\setminus B'$ to $j$. The existence of this partition implies  that agents in $B'$ only nominate agents in $B'$ so that $c_h \subset B'$ $\forall h \in B'$. Moreover, since each agent nominates $p$ agents we have that $|B'| \geq p+1$. We have thus proved that $B'$ is a bloc. 

It follows that if $B$ is a bloc which is not strongly connected, it contains another bloc of smaller size. Thus, since the minimal size of a bloc is $p+1$, for each bloc $B$ which fails to be strongly connected, there is a bloc contained in $B$ which is strongly connected.

\textbf{Uniqueness.} Assume by contradiction that for some profile $m$ there are two non-identical effective blocs $B^*$ and $B^{'*}$. Since each bloc consists of at least $p+1$ agents, $B^* \cap B^{'*} \neq \emptyset$. Thus, there is some agent $i$ such that $i \in B^* \cap B^{'*}$. By definition of a bloc, there is no path from $i$ to any $j \in B^* \setminus (B^* \cap B^{'*})$ since $i\in B^{'*}$. Likewise, there is no path from $i$ to any $h \in B^{'*} \setminus (B^* \cap B^{'*})$. By assumption, blocs $B^*$ and $B^{'*}$ are effective and, thus, strongly connected. It follows that there is a path between any two vertices of an effective bloc, reaching the desired contradiction.

$\boldsymbol{B^*=\cap_{B \in B^{m}}B}$\textbf{.} We have shown that each bloc which is not effective includes an effective bloc. We have also shown that the effective bloc is unique. The claim follows directly from the two observations.
\end{proof}

Lemma \ref{prop:effective} shows that in any profile $m$ with blocs, the intersection of the blocs is non-empty and is a bloc itself. Moreover, this intersection is strongly connected meaning that there is no agent which can be removed from it in such way that the profile still admits a bloc. This property has an important implication on the strategic behavior, as summarized by the next result: for any profile $m$ admitting a bloc, any agent in the effective bloc has a strategy $m'_i$ that allows her to break all blocs in $m$ (i.e. no bloc in $(m'_i,m_{-i})$).

\begin{lemma}\label{prop:deveffective}
For any profile $m$ admitting a bloc, any agent $i$ in the effective bloc $B^*$ has a strategy $m'_i$ such that $B^{(m'_i,m_{-i})}=\emptyset$.
\end{lemma}

\begin{proof}
Take some $m$ with $B^{m}\neq \emptyset$. W.l.o.g. assume that all blocs in $B^{m}$ are in favor of $x$ and consider some agent $i$ in the effective bloc $B^*$.

Observe first that there is a path from any $j \in B$  to $i$. Indeed, $|c_j|=p+1$ by definition, and therefore $c_j \cap B^*\neq \emptyset$ so that there is a path from $j$ to some agent $h\in c_j\cap B^*$. Moreover, since $B^*$ is strongly connected, there is a path from $h$ to $i$ since both belong to $B^*$: the existence of a path from $j$ to $i$ follows.

Consider a deviation $m'_i=(y,c_i)$, that is agent $i$ votes for $y$ instead of $x$ while keeping her nominations unchanged. After such deviation, $i$ cannot be a part of a bloc in favor of $x$ since she votes for $y$.

By definition of a bloc, for any bloc $B$ there is no path from the members of the bloc to the agents in $I \setminus B$. However, as stated before, since $i \in B^*$, there is a path to $i$ from any member of any bloc in $B^{m}$. Thus, there is no bloc in favor of $x$ in profile $(m'_i, m_{-i})$. 

Notice also that since $B^{m} \neq \emptyset$ and since $c_i$ includes only agents voting for $x$, there can be no bloc in favor of $y$ in $(m'_i,m_{-i})$. Thus, $B^{(m'_i,m_{-i})}=\emptyset$, ending the proof.\end{proof}

\subsection{Nash implementation \label{blocnash} }
 The main result of this section is as follows. 
\begin{proposition}\label{theorem:CFimpl}The Bloc formation mechanism Nash implements the majority rule.
\end{proposition}

The formal proof of Proposition \ref{theorem:CFimpl} can be found in the Appendix, but we provide some informal explanation in the following paragraphs. The existence of an equilibrium selecting the majority preferred option is simple. If there are at least $p + 1$ agents who prefer option $a$ and these agents vote for $a$ and nominate each other, this creates a bloc in favor of $a$. This profile is an equilibrium since no agent within the bloc wants to deviate (as they obtain their most preferred outcome) and no agent outside the bloc can alter the outcome (by definition).

 To discard the existence of an equilibrium with an outcome being a lottery with full support, observe that the weight with which each option wins is strictly increasing (1) in the number of votes it obtains from agents with positive weight and (2) in the number of nominations that agents voting for this option get. Therefore, it is optimal for agents to vote truthfully and nominate as many agents voting for their preferred option as possible, leading to the formation of a bloc in favor of option $a$. 

Finally, we can argue that no bloc can be formed in favor of $b$. Assume, by contradiction, that such a bloc exists. As shown in Section \ref{blocnetwork}, an effective bloc exists and includes some majority agent. Then, according to Lemma \ref{prop:deveffective}, a majority agent who is a member of the effective bloc can break all the blocs in favor of option $b$ in the profile leading to a lottery being the outcome. Such a deviation is profitable for a majority agent, contradicting the existence of an equilibrium in which a bloc in favor of option $b$ is formed. The following example illustrates the logic of the mechanism on this precise point.

\noindent \textit{Example 2:} Consider a profile $R$ with agents $1,2,3$ preferring $a$ and agents $4,5$  preferring $b$ so that $Maj(R)=a$. Remark that no equilibrium profile admits a bloc in favor of $b$. Indeed, let $m=(c,v)$ be the profile where each agent votes $b$ (i.e. $v_i =b\: \forall i$) and  nominations are as follows: $c_1=4,5$,  $c_2=1,3$, $c_3=1,2$, $c_4=1,5$ and $c_5=1,4$. The profile $m$ admits two blocs:  $\{1,4,5\}$ and $\{1,2,3,4,5\}$. The bloc $\{1,4,5\}$ is the effective one since one cannot find a smaller group that nominate each other while voting $b$. If any agent $i\in\{1,4,5\}$ deviates to $m'_i=(a,c_i)$, the profile $(m'_i,m_{-i})$ admits no bloc and the outcome is a lottery between $a$ and $b$. Since agent 1 prefers $a$ to $b$, she has a profitable deviation and thus the profile $m$ is not an equilibrium. 



\subsection{Bloc formation with an even number of voters \label{sec:evenBF}}

In this section we show that the implementation result for BF mechanism extends to the case when the number of agents is even, i.e. $n=2p$.

In order to incorporate the possibility of even number of agents we need to extend the notion of the majority rule. In this section $Maj$ is a social choice correspondence (SCC) such that for each preference profile $R$:

\begin{equation}
    Maj(R)=\begin{cases}
        a & |\{i \in I:aR_ib\}|\geq p+1,\\
        b & |\{i \in I:bR_ia\}| \geq p+1,\\
        \{a,b\} & \text{otherwise.}
    \end{cases}
\end{equation}

We extend the notion of Nash implementation and say that a mechanism $g$ Nash implements $Maj(R)$ if for all $R$ with $|Maj(R)|=1$, the outcome of any $NE^g(R)$ is $Maj(R)$; and if $Maj(R)=\{a,b\}$, the outcome of any $NE^g(R)$ is any lottery over the two options.

The definition of BF mechanism remains the same as introduced in Section \ref{simbloc}., notably, the minimal size of a bloc is still $p+1$.

\begin{lemma}
    With an even number of agents, the BF mechanism Nash implements the majority rule. For all $R \in \mathcal{R}$ such that $Maj(R)=A$, the unique equilibrium outcome is a lottery with equal weights.
\end{lemma}

\begin{proof}
    Remark first that for any profile $R$ where $Maj(R)$ is uniquely defined, the logic of the proof of Proposition \ref{theorem:CFimpl} applies verbatim. The only case that remains is the one when $Maj(R)=A$. We show that the unique equilibrium outcome in such case is a lottery that selects each of the options with equal probability.

    \textbf{Step 1: no bloc can be formed in equilibrium.} Consider some equilibrium $m$ and assume $B^{m} \neq \emptyset$. W.l.o.g. assume that a bloc is formed in favor of $a$. Then according to Lemma \ref{prop:effective} the effective bloc $B^*$ exists. Since $|B^*| \geq p+1$ there is some $i \in B^*$ such that $b R_i a$. Then according to Lemma \ref{prop:deveffective}, agent $i$ has a deviation $m'_i=(b,c_i)$ which breaks all the blocs in $B^{m}$ and thus switches the outcome to a lottery. Thus, the deviation is profitable for agent $i$, showing that $m$ is not an equilibrium.

\textbf{Step 2: no lottery which assigns higher probability to one of the options can be an equilibrium outcome.} By contradiction assume that a lottery $\eta(m)$ is an outcome for some equilibrium profile $m$ with $\eta^a(m)>1/2$. That means $\sum_{i\in I}\eta_{i}(m) \mathbbm{1}\{v_i=a\}>\sum_{i\in I}\eta_{i}(m) \mathbbm{1}\{v_i=b\}$, that is the weights of $a$-agents exceed the weights of $b$-agents. Then one of the following statements is true:

- There is some agent $i$ with $b R_i a$ such that $\eta_i(m)>0$ and $v_i=a$. In this case agent $i$ has a profitable deviation $m'_i=(b,c_i)$;

- There is some agent $i$ with $b R_i a$ such that $\exists j \in c_i$ with $v_j=a$ and $\exists h \notin c_i$ with $v_h=b$. That is there is an agent who prefers $b$ to $a$ but votes for $a$-agent when a $b$-agent is available. In this case agent $i$ has a profitable deviation $m'_i=(b, c_i \setminus \{j\} \cup \{h\})$.

- None of the above holds, but there is some agent $i$ with $b R_i a$ with $\eta_i(m)=0$ and $v_i=a$. Consider some agent $j$ with $a R_j b$. Since there is no agent who prefers $b$ to $a$, votes for $a$ and has positive weight it must be that there is $h \in c_j$ with $v_h=b$ (since $|c_j|=p$). Then agent $j$ has a profitable deviation $m'_j=(a,c_j \setminus \{h\} \cup \{i\})$ contradicting that $m$ is an equilibrium.
\end{proof}

\section{A sequential mechanism  \label{sec:majorityform}}

This section presents the Majority voting mechanism with Random confirmations (RC mechanism). This mechanism combines the features of the standard simultaneous and sequential majority procedures; this combination reduces the length of all equilibria compared to the sequential majority voting (with the shortest equilibrium including just 2 stages) while keeping the equilibrium outcome unique (in contrast to the one-shot majority voting).

The rules of the RC mechanism are as follows.

\noindent \textit{Voting stage:} Each agent $i$ votes for an option $v_i \in A$.

The profile of votes $v$ is publicly announced.  We can also relax this step of the mechanism by revealing only the shares of votes for each option which can be more suitable for practical use. We discuss this relaxation in Appendix C. \color{black}

\noindent \textit{Outcome of the Voting stage:} The option with most votes in $v$ is denoted the winner of the Voting stage.

\noindent \textit{Confirmation stage:} 
 A subset of agents of size $p+1$ is randomly chosen and ordered through a uniform draw. We denote the order by $\pi=(\pi_1, \ldots, \pi_{p+1})$. At each stage $t \in {\{1,\ldots,p+1\}}$, agent $\pi_t$ announces $Y$ or $N$. 


\noindent \textit{Outcome of the Confirmation stage::} 

Whenever an agent $\pi_t$ announces $Y$, the game ends, the outcome being the winner of the Voting stage. If no agent announces $Y$, this means that  all agents in $\{\pi_1, \ldots, \pi_{p+1}\}$ announce $N$. In this case, the outcome is the lottery $\beta(v)$ that assigns to each option its share of Voting stage votes, so that
$$\beta_a(v)=\frac{|\{i \in I : v_i=a\}|}{n} \: \text{ and } \beta_b(v)=1-\beta_a(v).$$

The purpose of this lottery is to give incentives to agents to vote for their most preferred option. Notice that the Confirmation stage announcements $Y$ and $N$ can be viewed as agreement and disagreement with the Voting stage outcome respectively. That is, if some agent in $\{\pi_1, \ldots, \pi_{p+1}\}$ agrees with the outcome being the winner of the Voting stage, this option is the outcome. On the other hand, if no one agrees, the outcome is the previously mentioned lottery.

We now establish the implementation under complete and incomplete information and discuss the extension to abstention.

\subsection{Complete information }

To provide a better understanding of the mechanism, we first present an example that demonstrates its logic before proceeding with a formal argument that encompasses the whole argument.

\medskip

\noindent \textit{Example 1}: Consider a preference profile $R$ with agents $1,2,3$ preferring $a$ to $b$ and agents $4,5$ preferring $b$ to $a$ so that $Maj(R)=a$. If all agents vote $b$ in the Voting stage (i.e. $v_i=b$ for all $i \in I$), the outcome  is $b$ independently of the Confirmation stage votes. Notice that in case the outcome was determined by simultaneous majority, such profile $v$ would be an equilibrium selecting a minority-preferred option. We now demonstrate that under RC mechanism there is always a voter who has a profitable deviation given such Voting stage profile. Consider agent 1's deviation from $v_1=b$ to $v'_1=a$ so that $(v'_1,v_{-1})=(a,b,b,b,b)$. After this deviation, the outcome depends on the votes of agents in $\{\pi_1,\pi_2,\pi_3\}$ in the Confirmation stage. If either agent 4 or agent 5 is in this set, $b$ is the outcome since both agents prefer $b$ to $a$ and, thus, their best response is to announce $Y$. Otherwise, the set $\{\pi_1, \pi_2,\pi_{3}\}$ equals $\{ 1, 2,3\}$ up to a permutation. Notice that such order $\{\pi_1, \pi_2,\pi_{3}\}$ occurs with strictly positive probability. The unique best response of any of these agents is to announce $N$ if $Y$ was not announced before. Thus, in the unique SPE of the Confirmation stage, $N$ is announced by all three agents and the outcome of the mechanism is a lottery that selects $a$ with probability 1/5. 
Therefore, by deviating from  $v_1$ to $v'_1$, agent 1 induces a lottery that assigns $a$ a positive probability; by \textbf{SD}, agent 1 prefers to deviate showing that any strategy profile in which $v_i=b$ for every agent $i$ cannot be an equilibrium.

Table \ref{tab:MF} illustrates this example. The left part represents the Voting stage profiles: unanimous in the first case, and after the deviation of agent 1 afterwards. The right part illustrates the SPE outcome of the Confirmation stage given the Voting stage profiles and the set $\{\pi_1, \pi_2,\pi_{3}\}$.

\begin{table}[h]
    \centering
    \begin{tabular}{c|c|c|c|c||c|c|c|c|c||c|}
          \multicolumn{5}{c||}{Voting stage} & \multicolumn{5}{c||}{Confirmation stage} & Outcome\\
         \hline
 1 & 2 & 3 & 4 & 5 & 1 & 2 & 3 & 4 & 5 & \\
        \hline
        \hline
        \multicolumn{11}{c}{\textbf{Unanimous vote for} $b$, $\{\pi_1,\pi_2,\pi_3\}=\{1,2,3\}$}\\
        \hline
         $b$ & $b$ & $b$ & $b$ & $b$ & $N$ & $N$ & $N$ & - & - & $b$ \\
         \hline
        \hline
        \multicolumn{11}{c}{\textbf{Deviation to $v'_1=a$}, $\{\pi_1,\pi_2,\pi_3\}=\{1,2,5\}$}\\
        \hline
         \textbf{a} \color{black} & $b$ & $b$ & $b$ & $b$ & $N$ & $N$ & - & - & $Y$ & $b$ \\
          \hline
        \hline
        \multicolumn{11}{c}{\textbf{Deviation to $v'_1=a$}, $\{\pi_1,\pi_2,\pi_3\}=\{1,2,3\}$}\\
        \hline
         \textbf{a} \color{black} & $b$ & $b$ & $b$ & $b$ & $N$ & $N$ & $N$ & - & - & $1/5a+4/5b$
    \end{tabular}
    \caption{Majority voting with Random confirmations}
    \label{tab:MF}
\end{table}
\color{black}

\color{black}

A similar logic to the one described in the example shows that at least $p+1$ agents who prefer the majority option are sincere in the Voting stage of \textit{any} equilibrium which leads to the implementation result, stated formally as follows. 
While the proof here is written with an odd number of agents, it can be extended to situations with an even number of them by properly modifying the mechanism (as detailed in appendix \ref{section:rceven}).

\begin{proposition}\label{pro:MVRCcom}
The RC mechanism subgame perfect implements the majority rule.
\end{proposition}

\begin{proof}
We start solving the game backwards from the Confirmation stage. If the profile $v$ is unanimous, the Confirmation stage does not affect the outcome. However, if $v$ is not unanimous, we denote by $x$ the winning option of the Voting stage and by $y$ the remaining option.

Next we consider agent $\pi_{p+1}$ and assume that no agent from $\{\pi_1, \ldots, \pi_{p}\}$ announced $Y$. If $xR_{\pi_{p+1}}y$, the unique best response for agent $\pi_{p+1}$ is $Y$; otherwise, it is $N$.

Moving on to agent $\pi_{p}$ and assuming no $Y$ was announced before, the best response for $\pi_{p}$ is:

- $Y$ if $xR_{\pi_p}y$ and  $yR_{\pi_{p+1}}x$;

- $N$ if $yR_{\pi_p}x$ and  $yR_{\pi_{p+1}}x$;

- $\{Y,N\}$ otherwise.

This logic can be extended to earlier agents in $\pi$ in the following way. For 
any agent $\pi_i$ with $i \leq p+1$ the best response in the Confirmation stage is:

- $Y$ if $xR_{\pi_i}y$ and  $yR_{\pi_j}x$ for all $i<j\leq p+1$;

- $N$ if $yR_{\pi_j}x$ for all $i \leq j\leq p+1$;

-$\{Y,N\}$ otherwise.

Then the SPE outcome of the Confirmation stage is the following one.

\begin{lemma}\label{lemma:RCcompl}
    For any non-unanimous profile $v$ of the Voting stage with $x$ being the winner, the SPE outcome of the Confirmation stage is:

- $x$ if $x R_i y$ for some $i \in \{\pi_1,\ldots, \pi_{p+1}\}$;

- the lottery $\beta(v)$ otherwise.
\end{lemma}













Consider the Voting stage assuming that $Maj(R)=a$ with $b$ being the minority preferred option. 
Remark that if $a$ wins the Voting stage, the equilibrium outcome is $a$. This follows from Lemma \ref{lemma:RCcompl} and the fact that for any order $\pi$, the subset $\{\pi_1, \ldots, \pi_{p+1}\}$ includes some agent who prefers $a$ to $b$.

We claim that in any equilibrium the score of $a$ in the profile $v$ is greater than or equal to $p+1$. Suppose, by contradiction, that this is not the case, and there are fewer than $p+1$ votes in favor of $a$ in $v$. Then there are only two possible cases in which $b$ wins with positive probability.

\noindent \textbf{Case 1:} $v$ is unanimous in favor of $b$. This logic of this case is analogous to  the one in Example 1. Consider some agent $i$ with $a R_i b$. If she deviates to $v'_i=a$, then if $\{\pi_1, \ldots, \pi_{p+1}\}$ contains only majority agents, the outcome is a lottery according to Lemma \ref{lemma:RCcompl} and it assigns positive probability to $a$. Thus, such deviation is profitable for agent $i$. It follows that the unanimous profile $v$ in favor of $b$ is not an equilibrium.

\noindent \textbf{Case 2:} There are some votes for $a$ in profile $v$, but less than $p+1$. This implies that there is some agent $i$ with $v_i=b$ while $aR_i b$. For this agent, deviating to $v'_i=a$ is profitable. Indeed, if after such deviation, $a$ is the winner of the Voting stage, then $a$ is the outcome. 
If this is not the case, $b$ is the winner of $(v'_i,v_{-i})$. By Lemma \ref{lemma:RCcompl}, if $\{\pi_1,\ldots,\pi_{p+1}\}$ includes only majority agents, which occurs with strictly positive probability, the outcome is the lottery $\beta(v'_i,v_{-i})$. Deviation by $i$ to $v'_i=a$ increases the probability of $a$ being selected in the lottery: 

    $$\beta_a(v)=\frac{|\{i\in I : v_i=a\}|}{n}<\frac{|\{i\in I : v_i=a\}|+1}{n}=\beta_a(v'_i,v_{-i}).$$
Therefore, due to \textbf{SD}, the agent $i$ finds profitable to deviate to $v'_i$.

Thus, we have eliminated all profiles $v$ in which less than $p+1$  agents vote $a$ as potential equilibria and this completes the proof.
\end{proof}

\subsection{Incomplete information }

We now prove that the RC mechanism implements the majority rule when we relax the assumption of complete information. It is assumed that each agent knows her own preference over the options and has some beliefs over the preferences of other agents.
For simplicity of the argument we assume that each agent believes that the preferences of other agents are $i.i.d.$ and assigns probability $q_a$ (resp. $1-q_a$) to each agent preferring $a$ to $b$ (resp. $b$ to $a$). Later we show that the $i.i.d.$ assumption can be relaxed without affecting the result. A strategy for an agent $i$ is a mapping $\sigma_i=(\sigma^1_i,\sigma^2_i)$ where $\sigma^1_i:\mathcal{R}\rightarrow A$ and $\sigma^2_i:\mathcal{R} \times H_{i} \rightarrow \{Y,N\}$ stand for the strategies in each of the stages with $h_1=v$ -the Voting stage profile, and $h_t \in \mathcal{R} \times \{Y,N\}^{t-1}$ for all $t \in \{2,\ldots,p+1\}$. I do not understand this notation, sorry. Where do we use it? This might help. We denote the conditional beliefs of agent $i$ about preferences of agent $j$ by $\mu_i(R_{j} \mid h_{i})$ with $\mu_i(aR_jb\mid \emptyset)=q_a$. A first-stage vote of agent $i$, $v^1_i$, is revealing given strategy $\sigma^1_i$ if $\mu_j(a R_i b \mid v^1_i)=1$ if $a R_i b$ and $\mu_j(a R_i b \mid v^1_i)=0$ otherwise for any agent $j \neq i$, that is if the preference of $i$ is uniquely determined given her first-stage vote.

\begin{proposition}\label{pro:MVRCincom}
    Under incomplete information, the RC mechanism implements the majority rule in PBE.
\end{proposition}

\begin{proof}
Fix any strategy profile $\sigma^1$, vote profile $v$ and order $\pi$ of agents. Denote by $x$ the winner of the Voting stage given $v$ and by $y$ the remaining option. Recall that in the Confirmation stage, an agent votes if and only if all previous votes are $N$, since otherwise (whenever a player announces $Y$) the game ends. 
 
The best response for the last mover, denoted agent $\pi_{p+1}$ is to vote $Y$ if $x R_{\pi_{p+1}} y$ and $N$ otherwise.

The best response of any agent $\pi_{t}$ with $t<p+1$ in the confirmation stage is as follows:

- $Y$ if $x R_{\pi_{t}} y$ and $\mu_{\pi_t}(y R_{\pi_{t+1}} x,\ldots,y R_{\pi_{p+1}} x) \mid h_{\pi_t})>0$, that is, agent $\pi_t$ assigns positive probability to the event that all agents in $\{\pi_{t+1}, \ldots, \pi_{p+1}\}$ prefer $y$ to $x$ (all successors have opposing preference);

- $N$ if $y R_{\pi_{t}} x$ and $\mu_{\pi_t}(y R_{\pi_{t+1}} x,\ldots,y R_{\pi_{p+1}} x) \mid h_{\pi_t})>0$, that is, agent $\pi_t$ assigns positive probability to the event that all successors have the same preference;

- $\{Y,N\}$ otherwise.

Note that based on our assumption on prior beliefs, if for some agent $\pi_t$, \break $\mu_{\pi_t}(y R_{\pi_{t+1}} x,\ldots,y R_{\pi_{p+1}} x) \mid h_{\pi_t})=0$ then, the same applies to the rest of the other agents. More precisely, this occurs only if some agent in $\{\pi_{t+1}, \ldots, \pi_{p+1}\}$ prefers $x$ to $y$ and her Voting stage strategy was revealing. It follows that, in case of non-revealing strategies for agents in $\{\pi_{t+1},\ldots,\pi_{p+1}\}$, agent $\pi_t$ strictly prefers to be truthful, i.e. to announce $Y$ if the winner of $v$ is her preferred option and $N$ otherwise. Thus, we can summarize the outcome of the Confirmation stage as follows. 


\begin{lemma}\label{lemma:RCincompl}
    For any strategy profile $\sigma^1$ and Voting stage profile $v$, the PBE outcome of the Confirmation stage is:

    - $x$ if $v_i=x$ for all $i \in I$,

    - $x$ if some agent in $\{\pi_1, \ldots, \pi_{p+1}\}$ prefers the winner at $v$,

    - the lottery $\beta(v)$ otherwise.
    
    \end{lemma}

    Consider the Voting stage of the mechanism, an agent $i \in I$ and an arbitrary profile $\sigma^1_{-i}$.  We now show that $i$ strictly prefers to vote for her most preferred option.

  Assume w.l.o.g. that $a R_i b$. Notice that there is no $v_{-i}$ for which $i$ strictly prefers to vote $b$. However, voting $b$ may be a best response if $i$ is indifferent between voting $a$ or voting $b$ for all possible realizations of preferences $R_{-i}$ and votes $v_{-i}$ of other agents given strategies $\sigma^1_{-i}$.

 Consider some realization $R_{-i}$ and $v_{-i}$ where $Maj(R)$ is not the unique option getting the most votes in $v_{-i}$ . In this case, $i$ strictly prefers to vote $a$ independently of whether  $Maj(R)=a$ or  $Maj(R)=b$. Indeed, when $i$'s vote is the $p+1^{th}$ in favor of $a$, by voting $a$ rather than $b$, she induces the outcome to be $a$ rather than a lottery (if $Maj(R)=a$) or a lottery which selects $a$ with higher probability under $v_i=a$ than $v_i=b$ (if $Maj(R)=b$).

 If for some strategy profile $\sigma^1_{-i}$, the option $Maj(R)$ gets the most votes in $v_{-i}$ for all realizations $R_{-i}$ and $v_{-i}$ that occur with positive probability, voting $b$ is a best response for $i$. We now show that no such strategy profile exists.

    Consider some profile $R_{-i}$ such that $p-1$ agents prefer $a$ to $b$ and the remaining $p+1$ agents prefer $b$ to $a$ so that $Maj(R)=b$. The profile $R_{-i}$ occurs with positive probability by assumption. Moreover, given $\sigma^1_{-i}$, any realization of $v_{-i}$ is such that the majority of agents in $I \setminus i$ vote for $b$ (since $Maj(R)=b$). Now consider a different profile $R'_{-i}$ such that $R_h=R'_h$ for all $h \in I \setminus \{i,j\}$ with $bR_ja$ and $aR'_jb$. That is, the profile $R'$ is such that $Maj(R)=a$ (since $a R_i b$) and the only difference with $R_{-i}$ is the preference of agent $j$. By assumption, in the profile $R'_{-i}$ for any realization $v_{-i}$ the majority of agents in $I \setminus \{i\}$ vote $a$. Notice, however, that each agent can condition her strategy only on her preference since this is the information available to agents in the Voting stage. Thus, for all agents in $I \setminus \{i,j\}$ the probability to vote for $a$ or $b$ remains the same when moving from $R_{-i}$ to $R'_{-i}$. Thus, the only change in votes occurs for agent $j$. Assume that either in $R_{-i}$ with $bR_ja$ or in $R'_{-i}$ with $aR'_jb$ agent $j$ randomizes, i.e. votes for $a$ and $b$ with positive probability. In this case, there is some profile $v_{-i}$ which occurs with positive probability under $R_{-i}$ and $R'_{-i}$. However, this contradicts the assumption that for any realization of $v_{-i}$  the majority preferred option obtains the majority of the votes. Thus, agent $j$ votes $a$ when $aR_jb$ and $b$ when $bR_ja$: she votes sincerely. Notice that agent $j$ was random so that the same logic applies to any agent in $I \setminus \{i\}$. Thus, $i$ is indifferent between voting $a$ or $b$ only if all other agents are truthful.

    Consider now a profile $R_{-i}$ such that exactly $p$ agents prefer $a$ to $b$ and $p$ remaining agents prefer $b$ to $a$. Since they are truthful there are $p$ votes for $a$ and $p$ votes for $b$. In this case if $v_i=a$ the outcome is $a$ as prescribed by Lemma \ref{lemma:RCincompl} whereas if $v_i=b$, the outcome is a lottery which assigns positive probability to $b$. Thus, $i$ strictly prefers to be truthful. This completes the proof.
    \color{black}
    

\end{proof}

Notice that our initial assumption on prior beliefs being $i.i.d.$ was unnecessarily demanding. If the prior beliefs satisfy the following weaker conditions, the result remains valid:

- Each agent assigns a positive probability to the event $T_a$ (the event where $p$
agents other than $i$ prefer $a$ to $b$ and $p$ agents prefer $b$ to $a$);

- For any subset of agents $I' \subset I$ and any agent $j \in I \setminus I'$, for any $x,y \in A$, $x \neq y$, $\mu_j(\forall i \in I', x R_i y \mid y R_j x)>0$. 

 The first assumption implies that if the rest of agents vote truthfully, an agent believes she is pivotal with positive probability since there might be exactly $p$ voters of each type. The second assumption ensures that at least some agent will not be indifferent between voting $Y$ and $N$ if her preferred alternative is the winner of the Voting stage independently of the strategy profile $\sigma_{-i}$.

\subsection{Abstention}

We now discuss an extension of the RC mechanism where we allow the agents to abstain. To distinguish from the original mechanism we call it RC mechanism with abstention. The possibility of abstention makes the strategic problem richer. Indeed, the abstention of many majority  agents can induce the victory of the minority and make agents indifferent between abstaining or voting for any of the options.\footnote{If one considers the simultaneous majority mechanism with abstention, there is a plethora of equilibria where the turnout of minority agents is larger than the one of majority agents and the minority preferred option wins.} In order to deal with abstention, we extend the definition of the mechanism as follows.

\vspace{.2cm}

\noindent \textit{Voting stage:} Each agent $i$ votes for an option or abstains, that is $v_i \in A \cup \{abs \}$. The profile of votes $v$ is publicly announced.  We call the option which gets most votes in $v$ the winner of the Voting stage. If no agent participates, i.e. $v_i=abs$ for all $i \in I$, or if the number of votes for $a$ is equal to the number of votes for $b$, the outcome is a lottery which assigns probability 1/2 to each of the options.

\noindent \textit{Confirmation stage:}

 A subset of agents of size $p+1$ is randomly chosen and ordered through a uniform draw (from all the agents independent of whether they participated or abstained in the Voting stage). We denote the order by $\pi=(\pi_1, \ldots, \pi_n)$. At each stage $t \in \{1,\ldots,p+1\}$, agent $\pi_t$ announces $Y$ or $N$ as long as $Y$ was not announced before.

\noindent \textit{Outcome:} 

If the two options are tied in the Voting stage (including the case in which all agents abstain) the outcome is the lottery which assigns probability 1/2 to each option. If there is a single winner in the Voting stage and some agent $\pi_t$ announces $Y$, the game ends, the outcome being the winner of the Voting stage. Finally, if all agents in $\{\pi_1, \ldots, \pi_{p+1}\}$ announce $N$ or abstain in the Confirmation stage the outcome is a lottery $\beta(v)$, which assigns to each option its share of first-stage votes, so that:
$$\beta_a(v)=\frac{|\{i \in I : v_i=a\}|}{|\{i \in I : v_i = a\}|+|\{i \in I : v_i = b\}|} \: \text{ and } \: \beta_b(v)=1-\beta_a(v).$$
\begin{proposition}\label{prop:RCabst}
The RC mechanism with abstention subgame perfect implements the majority rule in the presence of abstention.
\end{proposition}

The proof of the result can be found in the Appendix.
As a final comment on this mechanism, observe the existence of the following equilibrium. Consider a strategy profile in which only 2 agents vote in the Voting stage and both vote for the majority preferred option, and all agents vote $Y$ in the Confirmation stage whenever asked to vote. This is a subgame perfect equilibrium for any preference profile where the outcome is reached after only three votes and two steps. Indeed, no deviation is possible in the Voting stage since the winner is not altered by adding or substracting one vote. In the Confirmation stage, there is always a majority agent among the ones taking part so that she strictly prefers to announce $Y$. Thus, only two initial votes suffice to guarantee that the majority option is elected.

 \subsection{Supermajorities and RC mechanisms}

We now consider a final extension of the baseline model where both options are not treated symmetrically. The set of options consists of a status quo $sq$ and an alternative policy $x$, so that $A=\{sq,x\}$. Consider the  supermajority voting rule $Maj_k$ with:

\begin{equation}
    Maj_k(R)=
    \begin{cases}
        x & |\{i \in I : x R_i sq\}| \geq p+k \text{ and },\\
        sq & \text{otherwise.}
    \end{cases}
\end{equation}

where $1\leq k \leq p+1$. That is, for the alternative policy $x$ to be selected it needs to be preferred to status quo by at least $p+k$ agents. Notice that in case $k=1$, the rule $Maj_k$ is the simple majority rule considered in the rest of the paper, whereas when $k=p+1$, $Maj_k$ is the unanimity rule.

For each supermajority rule $Maj_k$, we provide an extension of the RC mechanism, denoted RC$_k$, that implements it in subgame-perfect equilibria. Its formal definition follows:

\noindent \textit{Voting stage:} Each agent $i$ votes simultaneously for one of the options, $v_i \in A$. The profile of votes $v$ and the winner (based on the supermajority rule $Maj_k$) are publicly announced.

\noindent \textit{Confirmation stage:} A subset of agents of size $\bar{t}$ is chosen and ordered through a uniform draw with $\bar{t}=p+k$ if $sq$ is the winner of the first stage and $\bar{t}=p+2-k$ otherwise. An order $\pi$ of agents is randomly chosen through a uniform draw. At each stage $t \in \{1,\ldots,\bar{t}\}$ agent $\pi_t$  announces $Y$ or $N$.

As in the baseline model, the mechanism ends at stage $t \leq \bar{t}$ if agent $\pi_t$ announces $Y$, the winner of the first stage being the outcome. If all agents in $\{\pi_1,\ldots, \pi_{\bar{t}}\}$ announce $N$ the outcome is the lottery with weights $\beta_a(v)$ and $\beta_b(v)$ given by the share of Voting stage votes.

The main difference with the benchmark $RC$ mechanism is that the number of agents selected for the Confirmation stage varies as a function of the winner of the Voting stage. For instance, in the case of the unanimity rule ($k=p+1$), the $RC_{p+1}$ mechanism only requires one agent in the Confirmation stage if the status quo wins (to be certain that at least some agent prefers $x$ to $sq$) whereas it requires that all agents take part in the Confirmation stage when $x$ wins initially (to be sure that all agents indeed have $x$ as their preferred option). As we now show, this simple modification of the $RC$ mechanism suffices to implement any supermajority rule.

\begin{proposition}
For any $1\leq k \leq p$, the   $RC_k$ mechanism subgame perfect implements the supermajority voting rule $Maj_k$.
\end{proposition}

\begin{proof}
   Notice that the logic of the Confirmation stage holds for any supermajority rule, thus, the result of Lemma \ref{lemma:RCcompl} applies. That is, in any profile $v$ with winner $x$, $x$ is the SPE outcome of the Confirmation stage if at least one agent in $\{\pi_1,\ldots,\pi_{\bar{t}}\}$ prefers $x$ to $y$ with $y\neq x$, and a lottery otherwise.

    Consider then the Voting stage. Assume first that some profile with $Maj_k(R)=sq$ (that is, less than $p+k$ agents prefer $x$ to $sq$) admits an equilibrium which selects $x$ with positive probability. Given the equilibrium outcome of the Confirmation stage discussed above it follows that in such equilibrium $x$ is the winner of the Voting stage. Thus, there is some agent $i$ such that $sq R_i x$ and $v_i=x$. Then, this agent has a profitable deviation since the set $\{\pi_1, \ldots, \pi_{\bar{t}}\}$ includes only agents who prefer $sq$ to $x$ with positive probability. Indeed, this holds since there at least $2p+1-(p+k-1)=p+2-k$ such agents  and $\bar{t}=p+2-k \leq p+k$ by definition. In this case, the outcome is a lottery that assigns higher probability to $sq$ under $(v'_i,v_{-i})$ than under the initial profile $v$.

    Notice that a symmetric logic applies if we consider a preference profile in which at least $p+k$ agents prefer $x$ to $sq$ and the existence of an equilibrium which selects $sq$ with positive probability, which completes the proof.
\end{proof}

\section{Conclusion \label{sec:conclusion}}

The main contribution of this work is the design 
of simple mechanisms that implement majority rule. The Bloc formation mechanism follows a different idea: agents need to avoid coordination problems by nominating each other and forming blocs.  The RC mechanism requests voters to confirm their vote after expressing a preference between two options. Additionally, our paper demonstrates that lotteries can be used to design mechanisms with fewer steps compared to existing methods, reinforcing the argument stated by \cite{abreu1991virtual} that lotteries can lead to more permissive implementation.

We acknowledge that future empirical research is needed to determine the welfare gains of the proposed mechanisms compared to traditional voting procedures. A natural extension of our work on the theoretical front would be to investigate the implementation of efficient rules subject to incentive compatibility, such as the weighted majority rules as characterized by \cite{azrieli2014pareto}. On the experimental front, we plan to explore the role of these mechanisms in participation games  (see \cite{xefteris2023participation}), as well as in other settings such as public good provision ones.

\vspace{.5cm}

\appendix

\noindent \textbf{\Large{Appendix}}
\section{Bloc formation mechanism}

Proof of Proposition \ref{theorem:CFimpl}.

 W.l.o.g. assume that any agent $i$ in $\{1,\ldots,p+1\}$ is such that $aR_i b$ so that $a$ is the majority-preferred option and $b$ the minority-preferred one. Any agent $i$ with $aR_i b$ is a majority agent. We need to prove that (A.) there is an equilibrium implementing $a$ and that (B.) any equilibrium selects $a$.

\vspace{.3cm}

\noindent A. Existence of an equilibrium selecting $a$. 

Consider the set $J=\{1,\ldots,p+1\}$ that consists only of majority agents. Take the strategy profile $m$ where for each $i\in J$, $v_i=a$ and $c_i\subset J\setminus \{i\}$ so that coalition $J$ forms a bloc in favor of $a$. It follows  that  $\chi_{BF}(m)=a$.
To see why $m$ is an equilibrium, remark that each agent in $J$ prefers $a$ to $b$ (and $a$ to any lottery with both $a$ and $b$ in its support by \textbf{SD}) and hence does not want to deviate. Each agent outside $J$ cannot affect the outcome since the bloc formed by $J$ is formed independently of the deviation of any agent outside $J$. This shows the existence of an equilibrium selecting $a$.

\vspace{.3cm}

\noindent B. Any equilibrium implements $a$.

\vspace{.2cm}

For the sake of clarity, we divide this part of the proof in two sections. In section B.1, we show that there is no bloc in favor of $b$ in equilibrium. In section B.2, we show that any strategy profile that leads to a full-support lottery cannot be an equilibrium, concluding the proof.

\vspace{.3cm}

\noindent B.1. No bloc in favor of  $b$ in equilibrium.

\vspace{.2cm}

Take any profile $m$ with a bloc $B$ in favor of $b$; hence $\chi_{BF}(m)=b$. The definition of a bloc means that at least $p+1$ agents vote for $b$ and nominate only agents in $B$. Consider the effective bloc $B^*$ which exists and is unique according to Lemma \ref{prop:effective}. Since $a$ is the majority option, there is some agent $i\in B^*$ with $v_i=b$ in the profile $m$ and $aR_i b$.

Assume that $m$ is an equilibrium. Suppose that agent $i$ deviates from $m_i=(b,c_i)$ to $m'_i=(a,c_i)$. This means that $B^*$ is not anymore an effective bloc in favor of $b$ in the profile $(m'_i,m_{-i})$. Moreover, since $B^*= \cap_{B \in B^{m}}B$, there is no other remaining bloc in the profile $(m'_i,m_{-i})$ as shown by Lemma \ref{prop:effective}; thus the outcome $\chi_{BF}(m'_i,m_{-i})$ is a lottery
with support $a$ and $b$ with $a$ being selected with positive probability since $\eta_{i}(m)>0$ and thus $\eta_i(m'_i,m_{-i})>0$ ($i$ was nominated by some other agent in $m$, being part of $B^*$). Thus, by $\textbf{SD}$, $m'_i$ is a profitable deviation for $i$ since it increases the probability of $a$ being selected, proving that $m$ is not an equilibrium.

\vspace{.3cm}

\noindent B.2. There is no equilibrium which selects $b$ with positive probability.

\vspace{.2cm}
Assume that there is some equilibrium $m$ where the outcome is a full-support lottery.

Notice that the following two statements hold for any equilibrium profile $m$ with the outcome being a lottery:

(1) any agent $i$ who is nominated ($\eta_i(m)>0$) is sincere. 

(2) any agent nominates the largest number of agents who announce her preferred option. 
In other words, if $aR_ib$ then $|\{j : v_j=a \text{ and } j \in c_i\}|=\min\{p ,\{h \in I : v_h=a\}\}$.

Indeed, (1) holds since  with $\eta_i(m)>0$ the vote of agent $i$ affects the final outcome, thus, voting sincerely increases the probability of $i$'s favorite option being selected. Statement (2) holds since the weight
$\eta^{x}(\cdot)$ is increasing in the sum of the weights of $x$-agents and each agent's weight strictly increases on the number of votes that she receives.

Given that (1) and (2) hold since $m$ is an equilibrium and that $B^m=\emptyset$, there is some majority agent which votes $b$ and is not nominated. Indeed, assume this is not the case and such agent does not exist. According to (1) all nominated agents vote sincerely. It follows from (2) then that all majority agents nominate only other majority agents who are also sincere. This means that a bloc in favor of $a$ exists contradicting $B^m=\emptyset$. Consider then some minority agent $j$, i.e. $bR_ja$. Since (1) holds, $c_j$ does not include any majority agent who votes $b$, that is $|\{h \in c_j : v_h=b\}|<p$. Then since $\eta^b(m)$ is increasing in the number of nominations of $b$-agents, agent $j$ has a profitable deviation: to nominate agent $i$ in $c_j$ rather than some $a$-agent. Formally, $m'_j=(b,c'_j)$ with $c'_j=(c_j \setminus \{h\}) \cup \{i\}$ for some $h$ with $v_h=a$. This contradicts $m$ being an equilibrium, and concludes the proof.

\section{RC mechanism with abstention}

Proof of Proposition \ref{prop:RCabst}

   In the Confirmation stage, all agents in $\{\pi_1,\ldots,\pi_{p+1}\}$ are indifferent between announcing $N$ or abstaining. Indeed, by construction the mechanism treats equally these announcements and, in the Confirmation stage, the best response does not depend on the previous announcements. 

    Assume first that the Voting stage admits a unique winner. In this case, the Confirmation stage outcome coincides with the one presented in Lemma \ref{lemma:RCcompl}.

    Assume now that both $a$ and $b$ are tied in the profile $v$. Thus,  the outcome is a lottery which assigns probability of 1/2 to each of the options. Therefore, all agents involved in the Confirmation stage are indifferent between all 3 possible announcements. Then, the counterpart of Lemma \ref{lemma:RCcompl} can be formulated as follows.

    \begin{lemma}\label{lemma:RCabst}
        For any non-unanimous profile $v$ of the Voting stage, the SPE outcome of the Confirmation stage is:

        - $x$ if $x$ is the unique winner in $v$ and $x R_i y$ for some $i \in \{\pi_1,\ldots,\pi_{p+1}\}$,

        - a lottery $\beta(v)$ if $x$ is the unique winner in $v$ and $y R_i x$ for any $i \in \{\pi_1,\ldots,\pi_{p+1}\}$,

        - a lottery which assigns equal probabilities to both options if $v$ does not admit a unique winner.
    \end{lemma}

    Consider now the Voting stage of the mechanism. We show that there is no equilibrium which selects $b$ (the minority preferred option) with positive probability. By contradiction, assume that such equilibrium exists.

    \textbf{Case 1:} The outcome is deterministic and selects $b$ with probability 1 for all orders $\pi$. In this case, given Lemma \ref{lemma:RCabst}, one of the following statements holds:

    - all participating majority agents vote $b$. If any of these agents deviates to $v'_i=a$, this is a profitable deviation since there is positive probability that only majority agents are selected at the Confirmation stage and, by Lemma \ref{lemma:RCabst}, the outcome in this case is a lottery;

    - no majority agent participates. Then for any majority agent $i$ with $v_i=abs$, deviating to $v'_i=a$ is profitable since it leads to a lottery as an outcome with positive probability.

    \textbf{Case 2:} The outcome is $b$ with positive probability. Notice, that if $a$ is the winner of the Voting stage, $a$ is the outcome for all possible orders $\pi$ since some majority agent is among the first $p+1$ agents at the Confirmation stage. Thus, if $b$ is selected with positive probability, she is the winner of $v$, or that $v_j=\emptyset$ for all $j \in I$. If $b$ is the winner of $v$ then there is some majority agent who either abstains or votes for $b$ in the Voting stage.

    - Assume $v_j=abs$ for all $j \in I$. In this case any agent has incentives to enter and vote for her favorite option since this option will be the outcome with only one agent present at the Voting stage.

    - Assume $b$ is the winner and there is some majority agent $i \in I$ with $v_i=b$. Then the deviation to $v'_i=a$ is profitable. Indeed, if after this deviation $a$ is the winner of $v$, $a$ is the outcome of the mechanism. Otherwise, the outcome is a lottery based on the Voting stage profile for any $\pi$. The deviation from $v_i$ to $v'_i=a$ increases the probability of $a$ in such lottery.

    - Assume $b$ is the winner and there is some majority agent $i \in I$ with $v_i=abs$. Then the deviation to $v'_i=a$ is profitable. Indeed, if after such deviation $a$ is the winner of the Voting stage, $a$ is the equilibrium outcome. 
If this is not the case, $b$ is the winner of $(v'-i,v_{-i})$. Thus, if some minority agent is in $\{\pi_1,\ldots,\pi_{p+1}\}$, $b$ is the outcome by Lemma \ref{lemma:RCabst}. However, if only majority agents are in $\{\pi_1,\ldots,\pi_{p+1}\}$, the outcome is the lottery $\beta(v'_i,v_{-i})$. Deviation by $v$ to $v'_i=a$ increases the probability of $a$ in such lottery: assuming that the number of votes for $a$ in $v$ is $n_a$ and the total number of the Voting stage participants is $n$ the probabilities are the following:

    $$\beta_a(v)=\frac{n_a}{n}<\frac{n_a+1}{n+1}=\beta_a(v'_i,v_{-i}).$$

    This concludes the proof.

\section{RC with shares revelation}

In this section we prove that the implementation results presented in Section 4 are robust to a simplification of RC mechanism where only the shares of votes for alternatives are revealed at the end of the voting stage. 

Note first that this relaxation of the mechanism does not affect the logic of the complete information stage, that is the proof of Proposition 1 remains correct. Indeed, since we work with complete information, the agents cannot extract any additional information from knowing the entire profile $v$ compared to knowing only the shares of votes for different options.

With incomplete information, however, the agents can extract more information about the other agents' preferences from the full profile $v$ rather than from the shares. In other words, the validity of Lemma \ref{lemma:RCincompl} is not guaranteed to hold. In what follows we show that this is indeed the case. For simplicity we continue to denote by $x$ the majority winner of the Voting stage and by $y$ the remaining option.

\begin{claim}
If in the RC mechanism only vote shares are revealed, for any strategy profile $\sigma^1$ and Voting stage profile $v$, the PBE outcome of the Confirmation stage is:

\begin{itemize}
    \item[--] $x$ if $v_i=x$ for all $i \in I$,
    \item[--] $x$ if $xR_iy$ for some $i \in \{\pi_1,\ldots,\pi_{p+1}\}$,
    \item[--] the lottery $\beta(v)$ otherwise.
\end{itemize}
\end{claim}

Assume some non-unanimous profile $v$ (so the agents know that the shares of both options are positive) and consider, firstly, agent $\pi_{p+1}$. It is the last agent to cast a vote, thus, she has the  information necessary to determine the outcome. Thus, her strategy remains unchanged, she votes $Y$ is $x R_{\pi_{p+1}}y$ and $N$ otherwise.

Consider now the agent $\pi_{p}$. There are two possible cases:

\textbf{Case 1: } Agent $\pi_p$ knows the preference of $\pi_{p+1}$ based on the strategy profile $\sigma^1$, on $|\{i \in I:v_i=a\}|$ and on the fact that no predecessor voted $Y$ in the Confirmation stage. \footnote{Note that there exist combinations of strategy profiles and Voting stage vote shares such that the knowledge of $\pi_{p+1}$ is possible: for instance, if all agents vote truthfully in the Voting stage, and all agent $\pi_p$ was the only vote to vote for the preferred alternative.} In this case, she strictly prefers to vote $Y$ if $xR_{\pi_p}y$ and the preference of $\pi_{p+1}$ are opposing, and to vote $N$ if $y R_{\pi_p} x$ and the preference of $\pi_{p+1}$ are the same. Otherwise, agent $\pi_p$ is indifferent between voting $Y$ and $N$.

\textbf{Case 2:} Agent $\pi_p$ does not know the preference of $\pi_{p+1}$ prior to her vote, that is she assigns positive probability to $x$ being both the preferred and the least preferred option of agent $\pi_{p+1}$. In this case $\pi_p$ strictly prefers to vote $Y$ if $xR_{\pi_p}y$ and $N$ if $y R_{\pi_p} x$, that is to vote according to her true preference to maximize the probability of her favorite option being elected.

Considering an arbitrary agent $i \in \{\pi_1,\ldots \pi_{p+1}\}$ she is indifferent between voting $Y$ and $N$ if

\begin{itemize}
    \item[--] $x R_i y$ and she assigns 0 probability to an event in which all agents to vote after her in the Confirmation stage have opposing preference;
    \item[--] $y R_i x$ and she assigns 0 probability to an event in which all agents to vote after her in the Confirmation stage have the same preference.
\end{itemize}

 Indeed, if none of the 2 cases holds, agent $i$ either prefers $x$ to $y$ and believes that she is the last person to vote with such preference with positive probability (so no subsequent agent will vote $Y$), or she prefers $y$ to $x$ and believes that none of the subsequent agents will vote $Y$ with positive probability (due to identical to $i$'s preference).

Otherwise, agent $i$ strictly prefers to vote according to her preference, that is to vote $Y$ if $x R_i y$ and to vote $N$ if $y R_i x$.

Note that since in equilibrium agents know the strategies of others, they assign probability 1 to an event only if the event takes place effectively. This completes the proof of the Lemma.

\section{RC with an even number of agents \label{section:rceven}}

The RC mechanism for the case of an even number of voters is extended by adding the following step: in case the Voting stage has a unique winner, then proceed to Confirmation stage as before; if both options get the same number of votes in the Voting stage, the outcome is a lottery which assigns equal weights to both options (no Confirmation stage needed).

\begin{claim}
    With an even number of agents, the RC mechanism subgame perfect implements the majority rule with any lottery being an equilibrium outcome for all $R \in \mathcal{R}$ such that $Maj(R)=A$.
\end{claim}

\begin{proof}
    First of all, note that the logic of Lemma \ref{lemma:RCcompl} holds independently of the number of agents whenever the winner of the Voting stage is well-defined (and the outcome is an equal weight lottery otherwise).

In the Voting stage, the argument is identical to the case of an even number of voters for each preference profile $R$ for which $Maj(R)$ is a singleton. For any $R$ such that $Maj(R)=A$, our extended implementation notion allows any lottery between the two options to be an equilibrium outcome. Thus, to show the existence of equilibrium is sufficient. One such possible equilibrium is the one in which each agent votes in the Voting stage according to her preference and the outcome is a lottery assigning equal weights to both options. To see that such equilibrium is not unique consider the Voting stage profile in which all agents vote $a$. Since among $p+1$ agents participating in the Confirmation stage there exists some agent $i$ such that $a R_i b$, the outcome is $b$. Note that no agent who prefers $b$ to $a$ has a profitable deviation, since $a$ is still the winner of the Voting stage after any such deviation.
\end{proof}

\bibliographystyle{aer}
\bibliography{biblio}

@book{qvortrup2005comparative,
  title={A comparative study of referendums: Government by the people},
  author={Qvortrup, Matt},
  year={2005},
  publisher={Manchester University Press}
}

@article{herrera2010quorum,
  title={Quorum and turnout in referenda},
  author={Herrera, H. and Mattozzi, A.},
  journal={Journal of the European Economic Association},
  volume={8},
  number={4},
  pages={838--871},
  year={2010},
  publisher={Oxford University Press}
}

@article{aguiar2016experimental,
  title={Experimental evidence that quorum rules discourage turnout and promote election boycotts},
  author={Aguiar-Conraria, L. and Magalh{\~a}es, P.C. and Vanberg, C.A.},
  journal={Experimental Economics},
  volume={19},
  pages={886--909},
  year={2016},
  publisher={Springer}
}

@article{abreu1991virtual,
  title={Virtual implementation in Nash equilibrium},
  author={Abreu, Dilip and Sen, Arunava},
  journal={Econometrica},
  pages={997--1021},
  year={1991},
  publisher={JSTOR}
}

@article{blais2019my,
  title={Was my decision to vote (or abstain) the right one?},
  author={Blais, Andr{\'e} and Feitosa, Fernando and Sevi, Semra},
  journal={Party Politics},
  volume={25},
  number={3},
  pages={382--389},
  year={2017},
  publisher={SAGE Publications Sage UK: London, England}
}

@article{azrieli2014pareto,
  title={Pareto efficiency and weighted majority rules},
  author={Azrieli, Yaron and Kim, Semin},
  journal={International Economic Review},
  volume={55},
  number={4},
  pages={1067--1088},
  year={2014},
  publisher={Wiley Online Library}
}

@article{borgers2014robust,
  title={Robust mechanism design and dominant strategy voting rules},
  author={B{\"o}rgers, Tilman and Smith, Doug},
  journal={Theoretical Economics},
  volume={9},
  number={2},
  pages={339--360},
  year={2014},
  publisher={Wiley Online Library}
}

@article{bochet2007nash,
  title={Nash implementation with lottery mechanisms},
  author={Bochet, Olivier},
  journal={Social Choice and Welfare},
  volume={28},
  number={1},
  pages={111--125},
  year={2007},
  publisher={Springer}
}

@article{benoit2008nash,
  title={Nash implementation without no-veto power},
  author={Beno{\^\i}t, Jean-Pierre and Ok, Efe A},
  journal={Games and Economic Behavior},
  volume={64},
  number={1},
  pages={51--67},
  year={2008},
  publisher={Elsevier}
}

@article{moore1988subgame,
  title={Subgame perfect implementation},
  author={Moore, John and Repullo, Rafael},
  journal={Econometrica},
  pages={1191--1220},
  year={1988},
  publisher={JSTOR}
}

@article{aghion2012subgame,
  title={Subgame-perfect implementation under information perturbations},
  author={Aghion, Philippe and Fudenberg, Drew and Holden, Richard and Kunimoto, Takashi and Tercieux, Olivier},
  journal={The Quarterly Journal of Economics},
  volume={127},
  number={4},
  pages={1843--1881},
  year={2012},
  publisher={MIT Press}
}

@article{vartiainen2007subgame,
  title={Subgame perfect implementation of voting rules via randomized mechanisms},
  author={Vartiainen, Hannu},
  journal={Social Choice and Welfare},
  volume={29},
  pages={353--367},
  year={2007},
  publisher={Springer}
}

@article{laslier2021solution,
  title={A solution to the two-person implementation problem},
  author={Laslier, Jean-Fran{\c{c}}ois and N\'uñez, Mat\'ias and Sanver, M. Remzi},
  journal={Journal of Economic Theory},
  volume={194},
  pages={105261},
  year={2021},
  publisher={Elsevier}
}

@article{maskin1999nash,
  title={Nash equilibrium and welfare optimality},
  author={Maskin, Eric},
  journal={The Review of Economic Studies},
  volume={66},
  number={1},
  pages={23--38},
  year={1999},
  publisher={Wiley-Blackwell}
}

@article{laslier2013incentive,
  title={An Incentive-Compatible Condorcet Jury Theorem},
  author={Laslier, Jean-Fran{\c{c}}ois and Weibull, J{\"o}rgen W},
  journal={The Scandinavian Journal of Economics},
  volume={115},
  number={1},
  pages={84--108},
  year={2013},
  publisher={Wiley Online Library}
  }

@article{nunez2019truth,
  title={Truth-revealing voting rules for large populations},
  author={N{\'u}{\~n}ez, Mat{\'\i}as and Pivato, Marcus},
  journal={Games and Economic Behavior},
  volume={113},
  pages={285--305},
  year={2019},
  publisher={Elsevier}
}

@article{azevedo2019strategy,
  title={Strategy-proofness in the large},
  author={Azevedo, Eduardo M and Budish, Eric},
  journal={The Review of Economic Studies},
  volume={86},
  number={1},
  pages={81--116},
  year={2019},
  publisher={Oxford University Press}
}

@article{echenique2022price,
  title={Price \& Choose},
  author={Echenique, Federico and N{\'u}{\~n}ez, Mat{\'\i}as},
  journal={arXiv preprint arXiv:2212.05650},
  year={2022}
}

@article{chen2023getting,
  title={Getting Dynamic Implementation to Work},
  author={Chen, Yi-Chun and Holden, Richard and Kunimoto, Takashi and Sun, Yifei and Wilkening, Tom},
  journal={Journal of Political Economy},
  volume={131},
  number={2},
  pages={285--387},
  year={2023},
  publisher={The University of Chicago Press Chicago, IL}
}

@techreport{xefteris2023participation,
  title={Participation games : Design and Experiments},
  author={Kirneva, Margarita and Núñez, Matías and Xefteris, Dimitrios},
  year={2023},
  institution={mimeo}
}

@article{xiong2021designing,
  title={Designing referenda: An economist's pessimistic perspective},
  author={Xiong, Siyang},
  journal={Journal of Economic Theory},
  volume={191},
  pages={105133},
  year={2021},
  publisher={Elsevier}
}

@article{fudenberg1991perfect,
  title={Perfect Bayesian equilibrium and sequential equilibrium},
  author={Fudenberg, Drew and Tirole, Jean},
  journal={journal of Economic Theory},
  volume={53},
  number={2},
  pages={236--260},
  year={1991},
  publisher={Elsevier}
}

\end{document}